\documentclass[a4paper,UKenglish,autoref]{lipics-v2019}

\usepackage[ruled,vlined,linesnumbered,commentsnumbered]{algorithm2e}
\usepackage{hyperref}

\title{Improved approximation ratio for covering pliable set families}

\titlerunning{Improved approximation ratio for covering pliable set families}

\author{Zeev Nutov}{The Open University of Israel}{nutov@openu.ac.il}
{https://orcid.org/0000-0002-6629-3243}{}
\authorrunning{Zeev Nutov}
\Copyright{Zeev Nutov}

\nolinenumbers


\acknowledgements{}

\begin{document}

\maketitle

\newcommand{\sem}    {\setminus}
\newcommand{\subs}   {\subseteq}
\newcommand{\empt}   {\emptyset}

\newcommand{\f}   {\frac}


\def\al   {\alpha}
\def\be {\beta}
\def\ga {\gamma}
\def\de   {\delta}
\def\eps {\epsilon}

\def\CC  {{\cal C}}
\def\FF   {{\cal F}}
\def\LL   {{\cal L}}
\def\TT   {{\cal T}}

\def\sfec    {{\sc Set Family Edge Cover}}
\def\nmcc  {{\sc Near Min-Cuts Cover}}
\def\ckecs  {{\sc Cap-$k$-ECSS}}                 
\def\fgc       {{\sc FGC}}                                 

\keywords{primal dual algorithm, pliable set family, approximation algorithm}

\begin{abstract}
A classic result of Williamson, Goemans, Mihail, and Vazirani [STOC 1993: 708--717]  
states that the problem of covering an uncrossable set family by a min-cost edge set 
admits approximation ratio~$2$, by a primal-dual algorithm with a reverse delete phase.
Recently, Bansal, Cheriyan, Grout, and Ibrahimpur [ICALP 2023: 15:1–15:19] showed that 
this algorithm achieves approximation ratio $16$ for a larger class of set families,
that have much weaker uncrossing properties. 
In this paper we will refine their analysis and show an approximation ratio of $10$.  
This also improves approximation ratios for several variants of the 
{\sc Capacitated $k$-Edge Connected Spanning Subgraph} problem.
\end{abstract}

\section{Introduction} \label{s:intro}

Let $G=(V,E)$ be  graph. 
For $J \subs E$ and $S \subs V$ let $\delta_J(S)$ denote the set 
of edges in $J$ with one end in $S$ and the other in $V \sem S$, and let $d_J(S)=|\delta_J(S)|$ be their number.
An edge set $J$ {\bf covers} $S$ if $d_J(S) \geq 1$.
Consider the following problem:

\begin{center}
\fbox{\begin{minipage}{0.98\textwidth} \noindent
\underline{\sc Set Family Edge Cover} \\ 
{\em Input:} \ \ A graph $G=(V,E)$ with edge costs $\{c_e:e \in E\}$ and a set family $\FF$ on $V$. \\ 
{\em Output:} A min-cost edge set $J \subs E$ such that $d_J(S) \geq 1$ for all $S \in \FF$.
\end{minipage}} \end{center}

In this problem the set family $\FF$ may not be given explicitly, but we will require that some queries related to $\FF$ can be 
answered in time polynomial in $n=|V|$. 
Specifically, following previous work, we will require that for any edge set $I$, 
the inclusion minimal members of the {\bf residual family  $\FF^I=\{S \in \FF:d_I(S)=0\}$} of $\FF$ 
(the family of sets in $\FF$ that are uncovered by $I$) can be computed in time polynomial in $n=|V|$.
We will also assume that $V,\emptyset \notin \FF$, as otherwise the problem has no feasible solution.


Agrawal, Klein and Ravi \cite{AKR} designed and analyzed a primal-dual approximation algorithm
for the {\sc Steiner Forest} problem, and showed that it achieves approximation ratio $2$. 
A classic result of Goemans and Williamson \cite{GW} from the early 90's shows by an elegant proof 
that the same algorithm applies for {\sc Set Family Edge Cover} with an arbitrary proper family $\FF$.
In fact, one of the main achievements of the Goemans and Williamson paper was defining 
a {\em generic class} of set families that models a rich collection of combinatorial optimization problems,
for which the primal dual algorithm achieves approximation ratio~$2$. 
Slightly later, Williamson, Goemans, Mihail, and Vazirani \cite{WGMV} 
further extended this result to the more general class of uncrossable families,
by adding to the algorithm a novel reverse-delete phase;
a set family $\FF$ is {\bf uncrossable} if 
$A \cap B,A \cup B \in \FF$ or $A \sem B,B \sem A \in \FF$ whenever $A,B \in \FF$.
They posed an open question of extending this algorithm to 
a larger class of set families and combinatorial optimization problems.
However, for 30 years, the class of uncrossable set families remained the most general generic class 
of set families for which the WGMV algorithm achieves a constant approximation ratio.

Recently, Bansal, Cheriyan, Grout, and Ibrahimpur \cite{BCGI} analyzed the performance of the WGMV algorithm 
for a class of set families that arise in several variants of capacitated network design problems. 
Recall that $\FF^I$ denotes the residual family of $\FF$ w.r.t. an edge set~$I$, 
and that two sets $A,B$ {\bf cross} if each of the four sets $A \cap B,V \sem (A \cup B),A \sem B,B \sem A$ is non-empty.

\begin{definition}
A set family $\FF$ is {\bf pliable} if for any $A,B \in \FF$ 
at least two of the sets $A \cap B,A \cup B, A\sem B,B \sem A$ belong to $\FF$.  
We say that $\FF$ is {\bf $\ga$-pliable} if it has the following additional property: \\
\noindent 
{\bf Property $(\gamma)$:} 
For any edge set $I$ and sets $S_1 \subset S_2$ in $\FF^I$,  
if an inclusion minimal set $C$ of $\FF^I$ crosses each of $S_1,S_2$,
then the set $S_2 \sem (S_1 \cup C)$ is either empty or belongs to $\FF^I$.
\end{definition}

Bansal, Cheriyan, Grout, and Ibrahimpur \cite{BCGI} showed that the WGMV algorithm 
achieves approximation ratio $16$ for $\gamma$-pliable families, and that Property $(\gamma)$ is essential -- 
without it the cost of the solution found by the WGMV algorithm can be
$\Omega(\sqrt{n})$ times the cost of an optimal solution.
They also considered applications of their result to several variants of capacitated network design problems, as follows. 

\medskip

\begin{center}
\fbox{\begin{minipage}{0.98\textwidth} \noindent
\underline{\nmcc} \\ 
{\em Input:} A graph $G_0=(V,E_0)$, an edge set $E$ on $V$ with costs $\{c_e:e \in E\}$, and an integer $k$. \\ 
{\em Output:} A min-cost edge set $J \subs E$ that covers the set family 
$\{\empt \neq S \subsetneq V: d_{E_0}(S) < k\}$.
\end{minipage}} \end{center}

\medskip

It is known that the set family in {\nmcc} is pliable, and \cite{BCGI} 
showed that it satisfies Property $(\gamma)$, thus obtaining a $16$-approximation. 

\medskip

\begin{center}
\fbox{\begin{minipage}{0.98\textwidth} \noindent
\underline{{\sc Capacitated $k$-Edge Connected Spanning Subgraph} ({\ckecs})} \\ 
{\em Input:} \ \ A graph $G=(V,E)$ with edge costs $\{c_e:e \in E\}$ and edge capacities $\{u_e:e \in E\}$, 
and an integer $k$. \\ 
{\em Output:} A mini-cost edge set $J \subs E$ such that $u(\de_J(S)) \geq k$ for all $\empt \neq S \subset V$.
\end{minipage}} \end{center}

\medskip

One can see that {\nmcc} is a particular case of {\ckecs}, when all edges in $E_0$ have cost $0$ and capacities $1$,
and other edges have capacity $k$. On the other hand, approximation ratio $\al$ for  {\nmcc} implies 
approximation ratio $\al \cdot \lceil k/u_{\min} \rceil$ for {\ckecs}, where $u_{\min}$ is the minimum capacity of an edge; see \cite{BCGI}.
This gives approximation ratio $16 \cdot \lceil k/u_{\min} \rceil$ for {\ckecs}.

Adjiashvili, Hommelsheim and M\"{u}hlenthaler \cite{AHM} 
introduced the following related problem, called $(k,q)$-{\sc Flexible Graph Connectivity}.
Suppose that there is a subset $U \subs E$ of ``unsafe'' edges, and we want to find 
the cheapest spanning subgraph $H$ that will be $k$-connected even if up to $q$ unsafe edge are removed. 
Let us say that a subgraph $H=(V,J)$ of $G$ is {\bf $(k,q)$-flex-connected} if
any cut $\de_H(S)$ of $H$ has at least $k$ safe edges or at least $k+q$ (safe and unsafe) edges. 
Namely, we require that $d_{H \sem U}(S) \geq k$ or $d_H(S) \geq k+q$ for all $\empt \neq S \subset V$.
Summarizing, we get the following problem.

\medskip

\begin{center}
\fbox{\begin{minipage}{0.96\textwidth} \noindent
\underline{{\sc $(k,q)$-Flexible Graph Connectivity} ($(k,q)$-{\fgc})} \\ 
{\em Input:} \ \ A graph $G=(V,E)$ with edge costs $\{c_e:e \in E\}$, $U \subs E$, and integers $k,q \geq 0$. \\  
{\em Output:} A min-cost subgraph $H$ of $G$ such that $H$ is $(k,q)$-flex-connected.
\end{minipage}} \end{center}

For various approximation algorithms for this problem see recent papers \cite{BCHI,BCGI,N-arx,CJ,B}.
Specifically, Bansal \cite{B} showed that if the problem of covering a $\ga$-pliable familily achieves approximation ratio $\al$, 
then $(k,3)$-{\fgc} admits approximation ratio $6+\al$, thus obtaining a $22$-approximation for $(k,3)$-{\fgc}.

\newpage

Another generalization of uncrossable families is considered in \cite{N-prd}.
A set family $\FF$ is {\bf semi-uncrossable} if for any $A,B \in \FF$ we have that
$A \cap B \in \FF$ and one of $A \cup B,A \sem B,B \sem A$ is in $\FF$, or $A \sem B,B \sem A \in \FF$. 
One can verify that semi-uncrossable families are sandwiched between uncrossable and $\ga$-pliable families. 
In \cite{N-prd} it is shown that the WGMV algorithm achieves the same approximation ratio $2$ for semi-uncrossable families
(the proof is essentially identical to that of \cite{WGMV} for uncrossable families), 
and are also given several examples of problems that can be modeled by semi-uncrossable families that are not uncrossable. 

We note that Bansal, Cheriyan, Grout, and Ibrahimpur \cite{BCGI} 
introduced several novel analysis techniques of the primal dual method, and derived a relevant Property $(\ga)$,
that enables to obtain a constant approximation ratio for pliable families without excluding known applications. 
Here we refine their analysis to improve their approximation ratio $16$ as follows. 

\begin{theorem} \label{t:main}
The {\sc Set Family Edge Cover} problem with a $\gamma$-pliable set family $\FF$ 
admits approximation ratio $10$.
\end{theorem}

Naturally, this also improves approximation ratios for several applications of $\ga$-pliable families discussed in \cite{BCGI,B} --
these improvements are summarized in Table~\ref{tbl:fgc}. 

\begin{table} [htbp] 
\begin{center}
\begin{tabular}{|l|l|l|}  \hline  
                                  & {\bf previous}                                                                & {\bf this paper}              
\\\hline \hline
{\nmcc}                    & $16$ \cite{BCGI}                                                           & $10$
\\\hline 
{\ckecs}                   & $16 \cdot \lceil k/u_{\min} \rceil$ \cite{BCGI}               &      $10 \cdot \lceil k/u_{\min} \rceil$
\\\hline
$(k,3)$-{\fgc}           & $22$ for $k$ even, $11+\eps$ for $k$ odd  \cite{B}  & $16$ for $k$ even       
\\\hline
\end{tabular}
\end{center}
\caption{Approximation ratios for {\nmcc}, {\ckecs}, and $(k,3)$-{\fgc}.}
\label{tbl:fgc}
\end{table}

\section{Proof of Theorem~\ref{t:main}} \label{s:t}

Here we prove Theorem~\ref{t:main}. We start with stating some simple properties of pliable families. 
One can see that if an edge $e$ covers one of the sets $A \cap B, A \cup B, A\sem B,B \sem A$ 
then it also covers one of $A,B$. This implies the following.

\begin{lemma} \label{l:res}
If $\FF$ is pliable or is $\ga$-pliable, then so is $\FF^I$, for any edge set $I$.
\end{lemma}

An {\bf $\FF$-core} is an inclusion minimal member of $\FF$; let $\CC_\FF$ denote the family of $\FF$-cores. 

\begin{lemma} \label{l:CS}
Let $\FF$ be a pliable set family and let $A \in \FF$ and $C \in \CC_\FF$. 
Then either $C \subs A$, or $C \cap A=\empt$, or $A \sem C,A \cup C \in \FF$.
Consequently, the members of $\CC_\FF$ are pairwise disjoint.   
\end{lemma}
\begin{proof}
Since $C \in \CC_\FF$,
we can have $C \cap A \in \FF$ only if $C \subs A$, and 
we can have $C \sem A \in \FF$ only if $C \cap A =\empt$.
In any other case, we must have  $A \sem C,A \cup C \in \FF$, since $\FF$ is pliable. 
\end{proof}

We now describe the algorithm. 
Consider the following LP-relaxation {\bf (P)} for {\sc Set Family Edge Cover} and its dual program {\bf (D)}:
\[ \displaystyle
\begin{array} {lllllll} 
&  \hphantom{\bf (P)} & \min         & \ \displaystyle \sum_{e \in E} c_e x_e & 
   \hphantom{\bf (P)} & \max         & \ \displaystyle \sum_{S \in \FF} y_S  \\
&      \mbox{\bf (P)} & \ \mbox{s.t.}  & \displaystyle \sum_{e \in \delta(S)} x_e  \geq 1 \ \ \ \ \ \forall S \in \FF \ \ \ \ \ \ \  &
       \mbox{\bf (D)} & \ \mbox{s.t.}  & \displaystyle \sum_{\delta(S) \ni e} y_S \leq c_e \ \ \ \ \ \forall e \in E \\
&  \hphantom{\bf (P)} &              & \ \ x_e \geq 0 \ \ \ \ \ \ \ \ \ \ \ \forall e \in E &
   \hphantom{\bf (P)} &              & \ \ y_S \geq 0 \ \ \ \ \ \ \ \ \ \ \ \ \forall S \in \FF 
\end{array}
\]

Given a solution $y$ to {\bf (D)}, an edge $e \in E$ is {\bf tight}
if the inequality of $e$ in {\bf (D)} holds with equality.
The algorithm has two phases.

{\bf Phase~1} starts with $I=\emptyset$ an applies a sequence of iterations.
At the beginning of an iteration, we compute the family $\CC=\CC_{\FF^I}$ of $\FF^I$-cores.
Then we raise the dual variables of the $\FF^I$-cores
uniformly (possibly by zero), until some edge $e \in E \sem I$ becomes tight, and add $e$ to $I$.
Phase I terminates when $\CC_{\FF^I}=\empt$, namely when $I$ covers $\FF$.

{\bf Phase~2} applies on $I$ ``reverse delete'', which means the following.
Let $I=\{e_1, \ldots, e_j\}$, where $e_{i+1}$ was added after $e_i$. 
For $i=j$ downto $1$, we delete $e_i$ from $I$ if $I \sem \{e_i\}$ still covers $\FF$.
At the end of the algorithm, $I$ is output.

It is easy to see that the produced dual solution is feasible, hence 
$\sum_{S \in \FF} y_S \leq {\sf opt}$, by the Weak Duality Theorem.
We prove that at the end of the algorithm 
$$
\sum_{e \in I} c(e) \leq 10 \sum_{S \in \FF} y_S \ .
$$
As any edge in $I$ is tight, the last inequality is equivalent to
$$
\sum_{e \in I} \sum_{\delta_l(S) \ni e} y_S \leq 10 \sum_{S \in \FF} y_S \ .
$$
By changing the order of summation we get:
$$
\sum_{S \in \FF} d_I(S) y_S \leq 10 \sum_{S \in \FF} y_S \ .
$$
It is sufficient to prove that at any iteration the increase at the left hand side is at most
the increase in the right hand side. 
Let us fix some iteration, and let $\CC=\CC_{\FF^I}$ be the family of cores among 
the members of $\FF$ not yet covered. The increase in the left hand side is 
$\varepsilon \cdot \sum_{C \in \CC} d_I(C)$, 
where $\varepsilon$ is the amount by which the dual variables were raised in the iteration, 
while the increase in the right hand side is 
$\varepsilon \cdot 10 |\CC|$. Consequently, it is sufficient to prove that
$\sum_{C \in \CC} d_I(C) \leq 10 |\CC|$.
As the edges were deleted in reverse order, the set $I'$ of edges in $I$ that were 
added after the iteration (and ``survived'' the reverse delete phase),
form an inclusion minimal edge-cover of the family $\FF'$ of members 
in $\FF$ that are uncovered at the beginning of the iteration.
Note also that $\bigcup_{C \in \CC} \delta_{I}(C) \subseteq I'$.
Hence to prove approximation ratio $10$, it is sufficient to prove the following purely combinatorial statement, 
in which due to Lemma~\ref{l:res} we can revise our notation to $\FF \gets \FF'$ and $I \gets I'$.

\begin{lemma} \label{l:2C}
Let $I$ be an inclusion minimal cover of a $\gamma$-pliable set family $\FF$ such that 
every edge in $I$ covers some $C \in \CC$. Then:
\begin{equation} \label{e:2C}
\sum_{C \in \CC} d_I(C) \leq 5 \cdot (2|\CC|-1) \ .
\end{equation}
Furthermore,  if $\FF$ is symmetric ($V \sem S \in \FF$ whenever $S \in \FF$) 
then $\sum_{C \in \CC} d_I(C) \leq 10 \cdot (|\CC|-1)$.
\end{lemma}

In the rest of this section we prove Lemma~\ref{l:2C}.
A set family $\LL$ is a {\bf laminar} if any two sets in $\LL$ are disjoint or one of them contains the other.
Let $I$ be an inclusion minimal edge cover of a set family $\FF$. 
We say that a set $S_e \in \FF$ is a {\bf witness set} for an edge $e \in I$ if 
$e$ is the unique edge in $I$ that covers $S_e$, namely, if $\delta_I(S_e)=\{e\}$.
We say that $\LL \subs \FF$ is a {\bf witness family} for $I$ if 
$|\LL|=|I|$ and for every $e \in I$ there is a witness set $S_e \in \LL$.
By the minimality of $I$, there exists a witness family $\LL \subs \FF$.
The following lemma was proved in \cite{BCGI}; 
we provide a proof for completeness of exposition.

\begin{lemma}[\cite{BCGI}] \label{l:witness}
Let $I$ be an inclusion minimal cover of a pliable set family $\FF$. 
Then there exists a witness family $\LL \subs \FF$ for $I$ that is laminar.
\end{lemma}
\begin{proof}
Let $A,B \in \FF$ be witness sets of edges $e,f \in I$, respectively. 
Note that no edge in $I \sem \{e,f\}$ covers a set from $A \cap B,A \cup B,A \sem B,B \sem A$,
as such an edge covers one of $A,B$, contradicting that $A,B$ are witness sets.   
Thus all possible locations of such $e,f$ are as depicted in Fig.~\ref{f:uncross}.
We claim that one of the sets $A \cap B, A \cup B, A \sem B, B \sem A$ is a witness set for one of $e,f$. 
This follows from the following observations.
\begin{itemize}
\item
$A \sem B \notin \FF$ in (a) and $B \sem A \notin \FF$ in (b), since these sets are not covered by $e,f$; 
in both cases, at least one of $A \cap B,A \cup B$ is in $\FF$, and it is a witness set for one of $e,f$.
\item
$A \cup B \notin \FF$ in (c), and $A \cap B \notin \FF$ in (d) since these sets are not covered by $e,f$; 
in both cases, at least one of $A \sem B,B \sem A$ is in  $\FF$, and it is a witness set for one of $e,f$.
\end{itemize}

This implies that if $\LL$ is a witness family for $I$, then for any $A,B \in \LL$ 
there are $A' \in \{A \cap B, A \cup B,A \sem B,B \sem A\}$ and $B' \in \{A,B\}$ such that
$\LL'=(\LL \sem \{A,B\})\cup \{A',B'\}$ is also a witness family for $I$.
W.l.o.g. we may assume that $B'=B$, so $\LL'$ is obtained from $\LL$ by replacing just one set $A$ by some 
$A' \in \{A \cap B, A \cup B,A \sem B,B \sem A\}$.

Let us say that two sets $A,B$ {\bf overlap} if 
each of the sets $A \cap B,A \sem B,B \sem A$ is non-empty. 
For $A \subs V$ let the {\bf overlapping number} of $A$ w.r.t. $\LL$ \ be 
$\be(A,\LL)= |\{B \in \LL : A,B \mbox{ overlap}\}|$.
Let $\be(\LL)=\sum_{A \in \LL} \be(A,\LL)$.
By the minimality of $I$ there exists a witness family for $I$. 
Let $\LL$ be a witness family for $I$ with $\be(\LL)$ minimal. 
We claim that $\LL$ is laminar.
Suppose to the contrary that there are sets $A,B \in \LL$ that overlap. 
Let $A',B',\LL'$ be as above, and w.l.o.g. assume that $B'=B$, so $\LL'=(\LL \sem \{A\}) \cup \{A'\}$
where  $A' \in \{A \cap B, A \cup B,A \sem B,B \sem A\}$.
By \cite[Lemma 23.15]{V}, if $A$ overlaps some $B \in \LL$, 
then for any $A' \in \{A \cap B, A \sem B,B \sem A, A \cup B\}$,
any $S \in \LL$ overlapped by $A'$ is also overlapped by $A$.
Consequently, $\be(A',\LL')<\be(A,\LL)$ (since $B$ is not overlapped by $A'$),
while $\be(S,\LL') \leq \be(S,\LL)$ for any other set in $\LL'$.
This implies $\be(\LL')<\be(\LL)$, contradicting the choice of $\LL$.
\end{proof}

\begin{figure} \centering \includegraphics{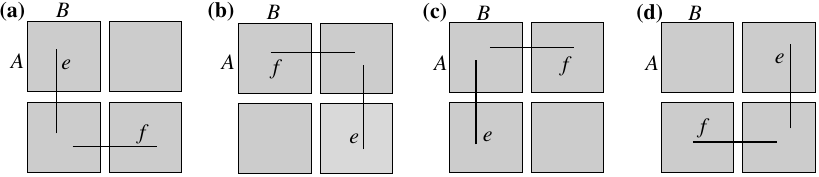}
\caption{Illustration to the proof of Lemma~\ref{l:witness}.}
\label{f:uncross} \end{figure}

Augment $\LL$ by the set $V$. 
Then $\LL$ can be represented by a rooted tree $\TT$ with node set $\LL$ and root $V$,
where the parent of $S$ in $\TT$ is the smallest set in $\LL$ that properly contains $S$.
The (unique) edge in $I$ that covers $S$ corresponds to the edge in $\TT$ from $S$ to its parent; 
a node $S$ of $\TT$ corresponds to the set of nodes in the set $S$ that do not belong to any child of $S$. 

\begin{definition} \label{d:chain}
A set $S \in \LL$ {\bf owns} a core $C \in \CC$ if $S$ is the inclusion-minimal set in $\LL$ that contains $C$.
We say that $S$ is {\bf hollow} if it owns no core.
A sequence $(S_0,S_1, \ldots, S_\ell)$ of sets in $\LL$ is called a {\bf hollow chain} (of length $\ell$)
if $S_1, \ldots, S_\ell$ are hollow, $S_\ell \neq V$, and $S_{i-1}$ is the unique child of $S_i$ in $\LL$, $i=1, \ldots,\ell$. 
For each $S_i$ let $a_ib_i$ be the unique edge in $I$ that covers $S_i$, where $a_i \in S_i$
(possibly $a_i=b_{i-1}$). 
\end{definition}

We will use for nodes of $\TT$ the the same terminology as for sets in $\LL$;
specifically, a node of $\TT$ is hollow if it represents a hollow set, and a hollow chain in $\TT$ is a path 
such that all internal nodes in the path correspond to hollow sets 
(note that each of these nodes has degree exactly $2$ in $\TT$).
We will need the following well known statement. 

\begin{lemma} \label{l:no-chain}
Let $T=(V_T,E_T)$ be a tree rooted at $r$ and let $R,W$ be a partition of  $V_T$ such that 
every $w \in W \sem \{r\}$ has at least $2$ children. 
Then $|E_T| \leq 2|R|-1$. Furthermore, if $r$ has at least $2$ children then $|E_T| \leq 2(|R|-1)$. 
\end{lemma}
\begin{proof}
No leaf of $T$ is in $W$ (note that $r$ is not a leaf), thus $|W| \leq |R|$; 
the worse case is when $R$ is the set of leaves of $T$, $r$ has $1$ child, and every node in $W \sem \{r\}$ has exactly $2$ children.
Consequently, $|E_T| \leq |W|+|R|-1 \leq 2|R|-1$. If $R$ has at least $2$ children 
then $|W| \leq |R|-1$ and we get  that $|E_T| \leq |W|+|R|-1 \leq 2|R|-2$.
\end{proof}

If $\TT$ has no hollow chain, then Lemma~\ref{l:no-chain} implies that 
$|\LL| \leq 2|\CC|-1$ and 
$|\LL| \leq 2(|\CC|-1)$ if $\FF$ is symmetric.
To see this, just let $W$ to be the set of hollow nodes in $\TT$ in Lemma~\ref{l:no-chain}, 
and note that if $\FF$ is symmetric then the root of $\TT$ has at least $2$ children.
Since the cores in $\CC$ are pairwise disjoint (by Lemmas \ref{l:res} and \ref{l:CS}), 
every edge contributes at most $2$ to $\sum_{C \in \CC} d_I(C)$, hence (since $|I|=|\LL|$) 
$\sum_{C \in \CC} d_I(C) \leq 2(2|C|-1)$, 
and $\sum_{C \in \CC} d_I(C) \leq 4(|C|-1)$ if $\FF$ is symmetric. 

Suppose now that every hollow chain has length $\ell$.  
Then the contribution of the edges $a_0b_0,\ldots, a_\ell b_\ell$ of each chain is at most $2(\ell+1)$.
If we ``shortcut'' every maximal hollow chain in $\TT$, 
we obtain a tree with at most $2|\CC|-1$ edges.
However, every such edge might be a shortcut of a hollow chain,
and thus may contribute $2(\ell+1)$ to $\sum_{C \in \CC} d_I(C)$.
As the number of edges after the shortcuts is at most $2|\CC|-1$, we get $\sum_{C \in \CC} d_I(C) \leq 2(\ell+1)(2|\CC|-1)$.
Bansal et. al. \cite{BCGI} showed that the maximum possible length of a hollow chain is $3$,
which gives the bound $\sum_{C \in \CC} d_I(C) \leq 8 \cdot (2|\CC|-1) < 16 \cdot |\CC|$.

Let $U = \cup_{C \in \CC} C$. To improve this bound of \cite{BCGI},
we will show that for any hollow chain with $\ell=2,3$,
among the relevant nodes $a_0,b_0,\ldots, a_\ell,b_\ell$ at most $3$ belong to $U$.
This reduces the bound on the contribution of every hollow chain from $8$ to $5$,
and gives the bounds 
$\sum_{C \in \CC} d_I(C) \leq 5 \cdot (2|\CC|-1)$ and 
$\sum_{C \in \CC} d_I(C) \leq 5 \cdot 2(|\CC|-1)$ if $\FF$ is symmetric.
Specifically, in the next section we will prove the following. 

\begin{figure} \centering \includegraphics{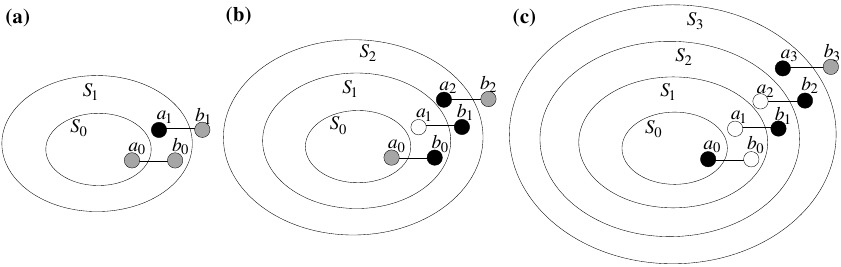}
\caption{Illustration to Lemma~\ref{l:main}.
Black nodes are in $U$, white nodes are not in $U$, while gray nodes may or may not be in $U$.}
\label{f:2} \end{figure}

\begin{lemma} \label{l:main}
Let $(S_0,S_1, \ldots, S_\ell)$ be a hollow chain and let $a_ib_i$ be as in Definition~\ref{d:chain}.  
Then $\ell \leq 3$ and the following holds, see Fig.~\ref{f:2}:
\begin{enumerate}[(a)]
\item
If $a_1 \in U$ then $\ell=1$.
\item
If $b_0 \in U$ then $\ell \leq 2$; if $\ell=2$ then $a_1 \notin U$ and $b_0,b_1,a_2$ belong to the same core.
\item
If $b_1 \in U$ then $\ell \leq 3$; if $\ell=3$ then $a_1,b_0,a_2 \notin U$ and $b_1,b_2,a_3$ belong to the same core.  
\end{enumerate}
\end{lemma}

Let us show that Lemma~\ref{l:main} implies Lemma~\ref{l:2C}. 
First, note that at least one of $a_1,b_1$ is in $U$, since every edge covers some core $C \in \CC$, 
by the assumption in Lemma~\ref{l:2C}. Thus $\ell \leq 3$. 
In the tree representation $\TT$ of $\LL$ let us ``shortcut'' all maximal hollow chains, see Fig.~\ref{f:3}.
This means that we replace the chain -- the edges of the chain 
(that correspond to edges $a_0b_0,\ldots,a_\ell b_\ell$) and 
the nodes that correspond to sets $S_1,\ldots , S_\ell$, 
by a new ``shortcut edge'' between $S_0$ and $S_\ell$
(for illustration, in Fig.~\ref{f:3} we assume that this new edge is $a_0b_\ell$).
This operation also can be viewed as identifying $S_0$ with $S_\ell$, by removing the nodes in $S_\ell \sem S_0$ 
(and adding to $I$ the edge $a_0b_\ell$).
Every such shortcut edge contributes at most $5$ (at most $4$ if $\ell=1$) 
to $\sum_{C \in \CC} d_I(C)$, by Lemma~\ref{l:main}.
Furthermore, the obtained tree has no hollow chains.  
By Lemma~\ref{l:no-chain}, the total number of edges in the obtained tree is at most $2|\CC|-1$, and at most $2(|\CC|-1)$ if $\FF$ is symmetric.
Each edge of $\TT$ contributes to $\sum_{C \in \CC} d_I(C)$ 
at most $2$ if it is an ordinary edge and at most $5$ if it is a shortcut edge. 
Consequently,  $\sum_{C \in \CC} d_I(C) \leq 5 \cdot (2|\CC|-1)$, and 
$\sum_{C \in \CC} d_I(C) \leq 10 \cdot (|\CC|-1)$ if $\FF$ is symmetric. 

This concludes the proof of Theorem~\ref{t:main}, provided that we will prove Lemma~\ref{l:main}, 
which we will do in the next section.

\begin{figure} \centering \includegraphics{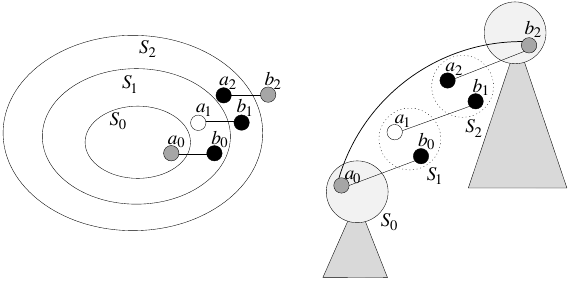}
\caption{Illustration of a shortcut of a hollow chain with $\ell=2$.
Black nodes are in $U$, white nodes are not in $U$, while gray nodes may or may not be in $U$.}
\label{f:3} \end{figure}

\section{Proof of Lemma~\ref{l:main}} \label{s:main}
 
To prove Lemma~\ref{l:main} we will apply Property $(\ga)$ on the family $\FF$ in Lemma~\ref{l:2C}
(that is in fact a residual family of the original family that we want to cover), so we restate here 
Property $(\ga)$ only for the current family $\FF$. 


\medskip

\noindent
{\bf Property $(\ga)$:} 
For any sets $S_1 \subset S_2$ in $\FF$,  
if an inclusion minimal set $C$ of $\FF$ crosses each of $S_1,S_2$,
then the set $S_2 \sem (S_1 \cup C)$ is either empty or belongs to $\FF$.

\medskip

In what follows, note that every $C \in \CC$ is owned by some some set in $\LL$, since $V \in \LL$. 
The next two Lemmas is a preparation for using Property $(\ga)$.

\begin{lemma} \label{l:contra}
Let $A$ be the child of a hollow set $S$. Then $S \sem A$ contains no set in $\FF$. 
\end{lemma}
\begin{proof}
Suppose to the contrary that $S \sem A$ contains a set $B \in \FF$.
Then $B$ contains some core $C \in \CC$. Let $S_C \in \LL$ be the set that owns $C$. 
Then we must have $S_C \subs S$ and $S_C \cap A=\empt$. 
Consequently, $S$ has at least $2$ children in $\LL$, contradicting that $S$ is hollow.
\end{proof}

\begin{lemma} \label{l:long-core}
Let $C \in \CC$ and let $1 \leq i \leq \ell$. 
If $(C \cap S_i) \sem S_0 \neq \empt$ then $C$ and $S_i$ cross. 
\end{lemma}
\begin{proof}
We need to show that  $C \cap S_i, S_i \sem C, C \sem S_i$ are nonempty and that $S_i \cup C \neq V$.  
\begin{itemize}
\item
$C \cap S_i \neq \empt$ by the assumption. 
\item
$S_i \sem C \neq \empt$ since $C$ does not contain $S_0$, hence there is $v \in S_0 \sem C \subset S_i \sem C$. 
\item
$C \cup S_i \neq V$ since $C \cup S_i \in \FF$ by Lemma~\ref{l:CS} and since $V \notin \FF$.
\end{itemize}
It remains to show that  $C \sem S_i \neq \empt$.
Suppose to the contrary that $C \subs S_i$. 
Let $S_C$ be the set that owns $C$. 
Since $C \subs S_i$, $S_C$ is a descendant of $S_i$.
Since $C \sem S_0 \neq \empt$, $S_C$ is not a descendant of $S_0$. 
Thus $S_C$ must be one of the sets $S_1, \ldots, S_i$, which is impossible, since all these sets are hollow.
 \end{proof}

Now we will use Property $(\ga)$.
Note that since the cores are pairwise disjoint (Lemma~\ref{l:CS}), 
then by the assumption that every edge covers some $C \in \CC$ (Lemma~\ref{l:2C}),
for any edge $uv \in I$ the following holds: 
$|\{u,v\} \cap U| \geq 1$ and 
$|\{u,v\} \cap C| \leq 1$ for any $C \in \CC$.

\begin{lemma} \label{l:a}
Let $(S_0,S_1, \ldots, S_\ell)$ be a hollow chain and let $a_ib_i$ be as in Definition~\ref{d:chain}.  
\begin{enumerate}[(i)]
\item
If $a_i \in U$ for some $i \geq 1$ then $\ell=i$.
\item
If $b_{i-1} \in U$ for some $i \geq 0$ then $\ell \leq i+1$ and if $\ell=i+1$ then 
$b_{i-1},b_i,a_{i+1}$ belong to the same core.
\end{enumerate}
\end{lemma}
\begin{proof}
For part (i), suppose to the contrary that  $S_{i+1}$ exists, see Fig.~\ref{f:4}(i).
Let $C \in \CC$ be such that $a_i \in C$. 
By Lemma~\ref{l:long-core}, $C$ crosses each of $S_i,S_{i+1}$.
Note that $b_i \notin C$, hence the set $S_{i+1} \sem (S_i \cup C)$ is non-empty, 
and thus by Property $(\gamma)$ is in $\FF$.  This contradicts Lemma~\ref{l:contra}. 

For part (ii), suppose that $C_{i+1}$ exists and let $C \in \CC$ be such that $b_{i-1} \in \CC$, see Fig.~\ref{f:4}(ii).
By Lemma~\ref{l:long-core}, $C$ crosses each of $S_i,S_{i+1}$.
If $b_i \notin C$ or if $a_{i+1} \notin C$ then the set $S_{i+1} \sem (S_i \cup C)$ is non-empty, 
and thus by Property $(\gamma)$ is in $\FF$, contradicting Lemma~\ref{l:contra}. 
Thus $b_i,a_{i+1} \in C$. Since $a_{i+1} \in U$, we get by part~(i) that $\ell=i+1$.
\end{proof}

\begin{figure} \centering \includegraphics{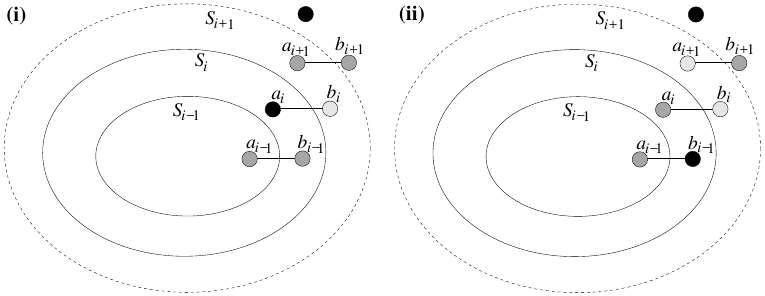}
\caption{Illustration to the proof of Lemma~\ref{l:a}.
(i) $a_i \in C$ and $b_i \notin C$. (ii) $b_{i-1} \in C$ and one of $b_i,a_{i+1}$ not in $C$.}
\label{f:4} \end{figure}

Lemma~\ref{l:main} easily follows from Lemma~\ref{l:a}.
\begin{enumerate}[(a)]
\item
If $a_1 \in U$ then by Lemma~\ref{l:a}(i) we have case~(a) of Lemma~\ref{l:main}, so assume that $a_1 \notin U$.
\item
If $b_0 \in U$ (and $a_1 \notin U$), then by Lemma~\ref{l:a}(ii) we have case~(b) of Lemma~\ref{l:main}.
\item
If $a_1,b_0 \notin U$ then $b_1 \in U$. 
Then either $\ell=2$, or $\ell=3$ and $b_2,b_2,a_3$ belong to the same core, so we have case~(c) of Lemma~\ref{l:main}.
\end{enumerate}

This concludes the proof of Lemma~\ref{l:main}, 
and thus also the proof of Lemma~\ref{l:2C} and Theorem~\ref{t:main} is complete.  



\end{document}